\documentclass[11pt, pdflatex]{article}

\usepackage{latexsym}
\usepackage{fullpage}
\usepackage{amssymb,amsmath}
\usepackage{graphicx}
\usepackage{color}
\usepackage{times}
\usepackage{algorithm}
\usepackage{algpseudocode}

\newcommand{\punt}[1]{}
\newcommand{\cmnt}[1]{}

\newtheorem{theorem}{Theorem}

\newtheorem{lemma}[theorem]{Lemma}

\newtheorem{property}[theorem]{Property}
\newtheorem{definition}[theorem]{Definition}

\newcommand{\secref}[1]{Section~\ref{sec:#1}}
\newcommand{\figref}[1]{Figure~\ref{fig:#1}}

\newcommand{\thmref}[1]{Theorem~\ref{thm:#1}}
\newcommand{\lemref}[1]{Lemma~\ref{lem:#1}}

\newcommand{\defref}[1]{Definition~\ref{def:#1}}

\newcommand{\propref}[1]{Property~\ref{prop:#1}}

\newcommand{\lineref}[1]{Line~\ref{lin:#1}}

\newcommand{\subsecref}[1]{SubSection~\ref{subsec:#1}}

\newcommand{\theqed}{$\Box$}

\newcommand{\qed}{\hspace*{\fill}\theqed\\\vspace*{-0.5em}}

\newcounter{linenumber}

\newcommand{\remove}[1]{}

\newcommand{\Wset}{\textit{Wset}}

\newcommand{\id}[1]{\mbox{\textit{#1}}}

\newcommand{\ignore}[1]{}

\newcommand{\op} {operation}
\newcommand{\termop} {terminal operation}

\newcommand{\comm}{\textit{committed}}
\newcommand{\aborted}{\textit{aborted}}
\newcommand{\txns}{\textit{txns}}

\newcommand{\evts}[1] {evts(#1)}
\newcommand{\ssch} {sub-history}

\newcommand{\lastw} {lastWrite}
\newcommand{\lwrite}[2] {#2.lastWrite(#1)}

\newcommand{\valid} {valid}
\newcommand{\validity} {validity}
\newcommand{\svvalid} {legal}

\newcommand{\legal} {legal}

\newcommand{\tobj} {t-object}

\newcommand{\trya}{\textit{tryA}}

\newcommand{\tseq} {t-sequential}
\newcommand{\opg}[2] {OPG(#1, #2)}
\newcommand{\ord}[1] {ord(#1)}
\newcommand{\ordfn} {ord}
\newcommand{\mv} {mv}
\newcommand{\rf} {rf}
\newcommand{\rt} {rt}
\newcommand{\mvto} {MVTO}
\newcommand{\memop} {memory operation}

\newcommand{\writei} {stm\_write}
\newcommand{\tryci} {stm\_tryC}
\newcommand{\tryc} {tryC}

\newcommand{\tcntr} {tCounter}
\newcommand{\livel} {liveList}
\newcommand{\begtrans} {begin\_tran}
\newcommand{\checkv} {check\_versions}
\newcommand{\init} {initialize}
\newcommand{\find} {find\_lts}
\newcommand{\instup} {ins\_tuple}
\newcommand{\remid} {removeId}
\newcommand{\gc} {gc}

\newcommand{\locko} {lockOrder}
\newcommand{\aco} {accessOrder}
\newcommand{\seq}[2] {linearize(#2, #1)}
\newcommand{\stl}[3] {#3.stl(#1, #2)}

\newcommand{\lts}[3] {#3.lts(#1, #2)}
\newcommand{\ltsn} {lts}
\newcommand{\vli}[3] {#3.vlist\_index(#1, #2)}

\algrenewcommand{\algorithmiccomment}[1]{$//$ #1}

\title{A TimeStamp based Multi-version STM Protocol that satisfies Opacity and Multi-Version Permissiveness \footnote{This work was in part supported by research funding provided by NetApp, Bangalore.}}

\author{Priyanka Kumar \\
priyanka@iitp.ac.in \\
Sathya Peri \\
sathya@iitp.ac.in \\
CSE Dept Indian Institute of Technology Patna\\
India \\ 
}

\begin{document}

\date{}
\maketitle
\thispagestyle{empty}

\begin{abstract}
Software Transactional Memory Systems (STM) are a promising alternative to lock based systems for concurrency control in shared memory systems. In multiversion STM systems, each write on a transaction object produces a new version of that object. The advantage obtained by storing multiple versions is that one can ensure that read operations do not fail. Opacity is a commonly used correctness criterion for STM systems. Multi-Version permissive STM system never aborts a read-only transaction. Although  many multi-version STM systems have been proposed, to the best of our knowledge none of them have been formally proved to satisfy opacity. In this paper we present a time-stamp based multiversion STM system that satisfies opacity and mv-permissiveness. We formally prove the correctness of the proposed STM system. We also present  garbage collection procedure which deletes unwanted versions of the transaction objects and formally prove it correctness.
\end{abstract}


\section{Introduction}
\label{sec:intro}
In recent years, Software Transactional Memory systems (STM) \cite{HerlMoss:1993:SigArch}, \cite{ShavTou:1995:PODC} have garnered significant interest as an elegant alternative for addressing concurrency issues in memory. STM systems take optimistic approach. Multiple transactions are allowed to execute concurrently. On completion, each transaction is validated and if any inconsistency is observed it is \emph{aborted}. Otherwise it is allowed to \emph{commit}. 

An important requirement of STM systems is to precisely identify the criterion as to when a transaction should be aborted/committed. Commonly accepted correctness criterion for STM systems is \emph{Opacity} proposed by Guerraoui, and Kapalka \cite{GuerKap:2008:PPoPP}. Opacity requires all the transactions including aborted to appear to execute sequentially in an order that agrees with the order of non-overlapping transactions. Opacity unlike traditional serializability \cite{Papad:1979:JACM} ensures that even aborted transactions read consistent values. 

With the increase in concurrency, more transactions may conflict and abort, especially in presence many long-running transactions which can have a very bad impact on performance \cite{AydAbd:2008:Serial:transact}. Perelman et al \cite{Perel+:2011:SMV:DISC} observe that read-only transactions play a significant role in various types of applications. But long read-only transactions could be aborted multiple times in many of the current STM systems \cite{herlihy+:2003:stm-dynamic:podc, dice:2006:tl2:disc}. In fact Perelman et al \cite{Perel+:2011:SMV:DISC} show that many STM systems waste 80\% their time in aborts due to read-only transaction. 

It was observed that by storing multiple versions of each object, multi-version STMs can ensure that read-only transactions do not abort. Maintaining multiple versions was first successfully used in databases. Since then, many STM systems have been developed that store multiple version of objects \cite{riegel+:sistm:transact:2006, riegel+:2006:LSA:DISC, Perel+:2011:SMV:DISC, perel+:2010:MultVer:PODC}. However storing multiple versions poses a difficulty of deleting unwanted object versions. It is necessary to regularly delete unused versions which otherwise could use a lot of memory. 

An important goal of STM system is to ensure that a transaction is not aborted when it does not violate correctness requirement. Many STM system however spuriously abort transactions \cite{attiyaHill:sinmvperm:tcs:2012} even not required. A \emph{permissive} STM \cite{Guer+:disc:2008} does not abort a transaction unless committing of it violates consistency requirements. A \emph{multi-version permissive} or \emph{mv-permissive} STM system \cite{perel+:2010:MultVer:PODC} never aborts a read-only transaction; it aborts an update transaction (i.e transaction that also writes) when it conflicts with other update transactions. 

Although many of the multi-version STM systems proposed in literature satisfy mv-permissiveness, no STM system to our knowledge has been formally of them are proved to satisfy opacity. In this paper, we propose a simple multi-version timestamp ordering STM system. We formally prove that our algorithm satisfies opacity and mv-permissiveness. To delete unwanted versions, we also give an algorithm for garbage collection and prove its correctness.

   
\vspace{1mm}
\noindent
\textit{Roadmap.} 
The paper is organized as follows. We describe our system model in \secref{model}. In \secref{gckt} we formally define the graph characterization for implementing the \emph{opacity}. In \secref{mvto}, we describe the working principle of MVTO protocol and its algorithm. In \secref{gar} we are collecting the garbage. Finally we conclude in \secref{conc}.

\section{System Model and Preliminaries}
\label{sec:model}

The notions and definitions described in this section follow the definitions of \cite{KuzSat:Corr:2012}. We assume a system of $n$ processes, $p_1,\ldots,p_n$ that access a
collection of \emph{objects} via atomic \emph{transactions}.
The processes are provided with  four \emph{transactional operations}: the
\textit{write}$(x,v)$ operation that updates object $x$ with value
$v$, the \textit{read}$(x)$ operation that returns a value read in
$x$, \textit{tryC}$()$ that tries to commit the transaction and
returns \textit{commit} ($c$ for short) or \textit{abort} ($a$ for
short), and \textit{\trya}$()$ that aborts the transaction and returns
$A$. The objects accessed by the read and write \op{s} are called as
\tobj{s}. For the sake of simplicity, we assume that the values written by all the transactions are unique. 

Operations \textit{write}, \textit{read} and \textit{\tryc}$()$ may
return $a$, in which case we say that the operations \emph{forcefully
abort}. Otherwise, we say that the operation has \emph{successfully}
executed.  Each operation is equipped with a unique transaction
identifier. A transaction $T_i$ starts with the first operation and
completes when any of its operations returns $a$ or $c$. 
Abort and commit \op{s} are called \emph{\termop{s}}. 

For a transaction $T_k$, we denote all its  read \op{s} as $Rset(T_k)$
and write \op{s} $Wset(T_k)$. Collectively, we denote all the \op{s}
of a  transaction $T_i$ as $\evts{T_k}$. 

\vspace{1mm}
\noindent
\textit{Histories.} A \emph{history} is a sequence of \emph{events}, i.e., a sequence of
invocations and responses of transactional operations. The collection
of events is denoted as $\evts{H}$. For simplicity, we only consider
\emph{sequential} histories here: the invocation of each transactional
operation is immediately followed by a matching response. Therefore,
we treat each transactional operation as one atomic event, and let
$<_H$ denote the total order on the transactional operations incurred
by $H$. With this assumption the only relevant events of a transaction
$T_k$ are of the types: $r_k(x,v)$, $r_k(x,A)$, $w_k(x, v)$, $w_k(x,
v,A)$, $\tryc_k(C)$ (or $c_k$ for short), $\tryc_k(A)$, $\trya_k(A)$ (or $a_k$ for short). 
We identify a history $H$ as tuple $\langle \evts{H},<_H \rangle$. 

\begin{figure}[tbph]
\centerline{\scalebox{0.7}{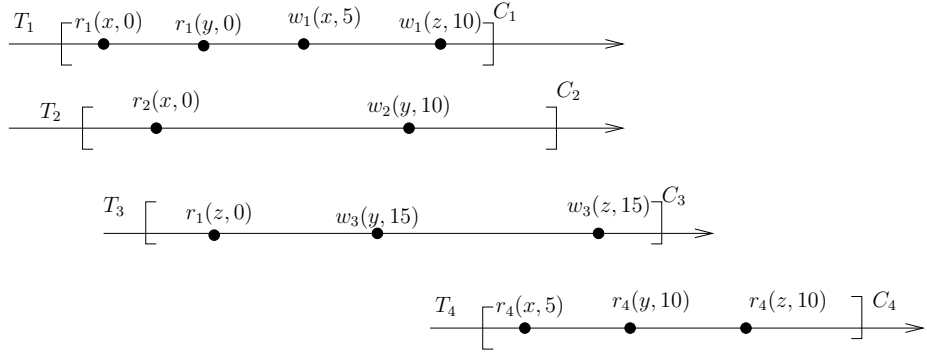}}
\caption{Pictorial representation of a History $H1$}
\label{fig:ex1}
\end{figure}

Let $H|T$ denote the history consisting of events of $T$ in $H$, and $H|p_i$ denote the history consisting of events of $p_i$ in $H$. We only consider \emph{well-formed} histories here, i.e., (1) each $H|T$ consists of  a read-only prefix (consisting of read operations only), followed by a write-only part (consisting of write operations only), possibly \emph{completed} with a $\tryc$ or $\trya$ operation\footnote{This restriction brings no loss of  generality~\cite{KR:2011:OPODIS}.}, and (2) each $H|p_i$ consists of a sequence of transactions, where no new
transaction begins before the last transaction completes (commits or a aborts). 

We assume that every history has an initial committed transaction $T_0$ that initializes all the data-objects with 0. The set of transactions that appear in $H$ is denoted by $\txns(H)$. The set of committed
(resp., aborted) transactions in $H$ is denoted by $\comm(H)$ (resp., $\aborted(H)$). The set of \emph{incomplete} (or \emph{live}) transactions in $H$ is denoted by $\id{incomplete}(H)$ ($\id{incomplete}(H)=\txns(H)-\comm(H)-\aborted(H)$). For a history $H$, we construct the \emph{completion} of $H$, denoted $\overline{H}$, by inserting $a_k$ immediately after the last event of every transaction $T_k\in\id{incomplete}(H)$.

\figref{ex1} shows a pictorial representation of a history $H1:r_1(x, 0) r_2(x, 0) r_1(y, 0) r_3(z, 0) w_1(x, 5) w_3(y, 15) \\
w_2(y, 10) w_1(z, 10) c_1 c_2 r_4(x, 5) r_4(y, 10) w_3(z, 15) c_3 r_4(z, 10)$. 

\vspace{1mm}
\noindent
\textit{Transaction orders.} For two transactions $T_k,T_m \in \txns(H)$, we say that  $T_k$ \emph{precedes} $T_m$ in the \emph{real-time order} of $H$, denote $T_k\prec_H^{RT} T_m$, if $T_k$ is complete in $H$ and the last event of $T_k$ precedes the first event of $T_m$ in $H$. If neither $T_k\prec_H^{RT} T_m$ nor $T_m \prec_H^{RT} T_k$, then $T_k$ and $T_m$ \emph{overlap} in $H$. A history $H$ is \emph{t-sequential} if there are no overlapping transactions in $H$, i.e., every two transactions are related by the real-time order. 

\cmnt {
We now define two order relations
on the set of transactions that are going to be instrumental in our
further definitions: real-time order, and deferred-update order.

For $T_k,T_m \in \txns(H)$, we say that  $T_k$ \emph{precedes} $T_m$
in the \emph{deferred-update order} of $H$, denote $T_k\prec_{H}^{DU} T_m$, if
$T_k$ contains a read $r_k(x,v)$, $T_m$ is committed,
$x\in\Wset(T_m)$, and $r_k(x,v)<_H c_m$. 
}

\cmnt {
\vspace{1mm}
\noindent
\textit{Sub-histories.} A \textit{sub-history}, $SH$ of a history
$H$ denoted as the tuple $\langle \evts{SH},$ $<_{SH}\rangle$ and is
defined as: (1) $<_{SH} \subseteq <_{H}$; (2) $\evts{SH} \subseteq
\evts{H}$; (3) If an event of a transaction $T_k\txns(H)$ is in $SH$ then all
the events of $T_k$ in $H$ should also be in $SH$. 
For a history
$H$, let $R$ be a subset of $txns(H)$, the transactions in
$H$. 
Then $\shist{R}{H}$ denotes  the \ssch{} of $H$ that is
formed  from the \op{s} in $R$. 
}

\vspace{1mm}
\noindent
\textit{Valid and legal histories.} Let $H$ be a history and $r_k(x, v)$ be a successful read {\op} (i.e $v \neq A$) in $H$. Then $r_k(x, v)$, is said to be \emph{\valid} if there is a transaction $T_j$ in $H$ that commits before $r_K$ and $w_j(x, v)$ is in $\evts{T_j}$. Formally, $\langle r_k(x, v)$  is \valid{} $\Rightarrow \exists T_j: (c_j <_{H} r_k(x, v)) \land (w_j(x, v) \in \evts{T_j}) \land (v \neq A) \rangle$. We say that $T_k$ and $T_j$ have a \textit{reads-from} relation in $H$. The history $H$ is \valid{}  if all its successful read \op{s} are \valid. 

We define $r_k(x, v)$'s \textit{\lastw{}} as the latest commit event $c_i$ such that $c_i$ precedes $r_k(x, v)$ in $H$ and $x \in \Wset(T_i)$ ($T_i$ can also be $T_0$). A successful read \op{} $r_k(x, v)$ (i.e
$v \neq A$), is said to be \emph{\svvalid{}} if transaction $T_i$ (which contains  $r_k$'s \lastw{}) also writes $v$ onto $x$. Formally, $\langle r_k(x, v)$ \text{is \svvalid{}} $\Rightarrow (\lwrite{r_k(x, v)}{H} = c_i) \land (w_i(x,v) \in \evts{T_i}) \land (v \neq A) \rangle$. The history $H$ is \svvalid{} if all its successful read \op{s} are \svvalid. Thus from these definitions we get that if $H$ is \svvalid{} then it is also \valid.

\vspace{1mm}
\noindent
\textit{Opacity.} We say that two histories $H$ and $H'$ are \emph{equivalent} if they have the same set of events. Now a history $H$ is said to be \textit{opaque} \cite{GuerKap:2008:PPoPP,tm-book} if $H$ is \valid{} and there exists a t-sequential legal history $S$ such that (1) $S$ is equivalent to $\overline{H}$ and (2) $S$ respects $\prec_{H}^{RT}$, i.e $\prec_{H}^{RT} \subseteq \prec_{S}^{RT}$. By requiring $S$ being equivalent to $\overline{H}$, opacity treats all the incomplete transactions as aborted. 

\cmnt {
Along the same lines, a \valid{} history $H$ is said to be
\textit{strictly serializable} if $\shist{\comm(H)}{H}$ is opaque.
Thus, unlike opacity, strict serializability does not include aborted
transactions in the global serialization order.
}

\section{Graph characterization of Opacity}
\label{sec:gckt}

To prove that a STM system satisfies opacity, it is useful to consider graph characterization of histories. The graph characterization described in this section is based on the characterization by Bernstein and Goodman \cite{BernGood:1983:MCC:TDS} and is slightly different from the characterisation of Gueraroui and Kapalka \cite{GuerKap:2008:PPoPP, tm-book}. 

Consider a history $H$ which consists of multiple version for each \tobj. Like \cite{BernGood:1983:MCC:TDS, GuerKap:2008:PPoPP, tm-book}, we use the notion of \textit{version order}. Given $H$ and a \tobj{} $x$, we define a version order for $x$ as any (nonreflexive) total order on all the versions of $x$ ever written by committed transactions in $H$. It must be noted that the version order may or may not be same as the actual order in which the version of $x$ are generated in $H$. A version order of $H$, denoted as $\ll_H$ is the union of the version orders of all the \tobj{s} in $H$. Using the notation that a committed transaction $T_i$ writing to $x$ creates a version $x_i$, a possible version order for $H1$ of \figref{ex1}, $\ll_{H1}$ is: $\langle x_0 \ll x_1 \rangle, \langle y_0 \ll y_2 \ll y_3 \rangle, \langle z_0 \ll z_1 \ll z_3 \rangle $. 

We define the graph characterisation based on a given version order. Consider a history $H$ and a version order $\ll$. Then a graph denoted as $\opg{H}{\ll}$ (opacity graph) can be defined. There is a vertex for each transaction $T_i$ in $\overline{H}$. The edges of the graph are of three kinds and are defined as follows:

\begin{enumerate}

\item \textit{\rt}(real-time) edges: If $T_i$ commits before $T_j$ starts in $H$, then there is an edge from $v_i$ to $v_j$. This set of edges are referred to as $\rt(H)$.

\item \textit{\rf}(reads-from) edges: If $T_j$ reads $x$ from $T_i$ in $H$, then there is an edge from $v_i$ to $v_j$. Note that in order for this to happen, $T_i$ must have committed before $T_j$ and $c_i <_H r_j(x)$. This set of edges are referred to as $\rf(H)$.

\item \textit{\mv}(multiversion) edges: The \mv{} edges capture the multiversion relations and is based on the version order. Consider a successful read \op{} $r_k(x,v)$ and the write \op{} $w_j(x,v)$ belonging to transaction $T_j$. Here, $r_k(x,v)$ reads $x$ from $w_j(x,v)$ (it must be noted $T_j$ is a committed transaction and $c_j <_H r_k$). Consider a committed transaction $T_i$ which writes to $x$, $w_i(x, u)$ where $u \neq v$. Thus the versions created $x_i, x_j$ are related by $\ll$. Then, if $x_i \ll x_j$ we add an edge from $v_i$ to $v_j$. Otherwise ($x_j \ll x_i$), we add an edge from $v_k$ to $v_i$. This set of edges are referred to as $\mv(H, \ll)$.
\end{enumerate}

Using this construction, the $\opg{H1}{\ll_{H1}}$ for history $H1$ and $\ll_{H1}$ is given above is shown in \figref{opg}. 

Given a history $H$ and a version order $\ll$, consider the graph $\opg{\overline{H}}{\ll}$. While considering the \rt{} edges in this graph, we only consider the real-time relation of $H$ and not $\overline{H}$. It can be seen that $\prec_H^{RT} \subseteq \prec_{\overline{H}}^{RT}$ but with this assumption, $\rt(H) = \rt(\overline{H})$. Hence,  we get the following property, 
\begin{property}
\label{prop:hoverh}
The graphs $\opg{H}{\ll}$ and $\opg{\overline{H}}{\ll}$ are the same for any history $H$ and $\ll$. 
\end{property}
Now we show the correctness of our graph characterization using the following lemmas and theorem. 

\begin{figure}[tbph]
\centerline{\scalebox{0.7}{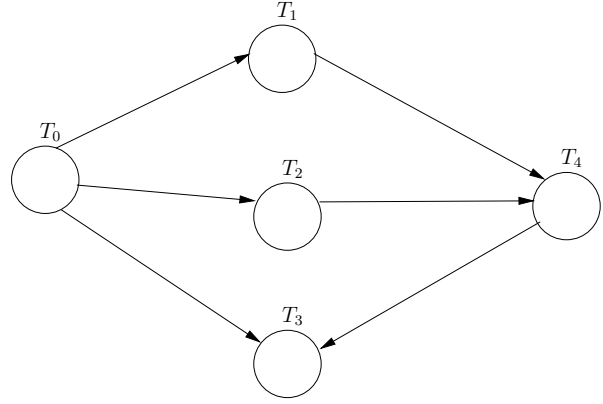}}
\caption{$\opg{H1}{\ll_{H1}}$}
\label{fig:opg}
\end{figure}

\begin{definition}
\label{def:seqver}
For a \tseq{} history $S$, we define a version order $\ll_S$ as follows: For two version $x_i, x_j$ created by committed transactions $T_i, T_j$ in $S$, $\langle x_i \ll_S x_j \Leftrightarrow T_i <_S T_j \rangle $. 
\end{definition}
Now, consider the following lemmas,

\begin{lemma}
\label{lem:seracycle}
Consider a \legal{} \tseq{} history $S$. Then the graph $\opg{S, \ll_S}$ is acyclic.
\end{lemma}

\begin{proof}
We numerically order all the transactions in $S$ by their real-time order by using a function \textit{\ordfn}. For two transactions $T_i, T_j$, we define $\ord{T_i} < \ord{T_j} \Leftrightarrow T_i <_S T_j$. Let us analyse the edges of $\opg{S, \ll_S}$ one by one: 
\begin{itemize}
\item \rt{} edges: It can be seen that all the \rt{} edges go from a lower \ordfn{} transaction to a higher \ordfn{} transaction. 

\item \rf{} edges: If $T_j$ reads $x$ from $T_i$ in $S$ then $T_i$ is a committed transaction with $\ord{T_i} < \ord{T_j}$. Thus, all the \rf{} edges from a lower \ordfn{} transaction to a higher \ordfn{} transaction.

\item \mv{} edges: Consider a successful read \op{} $r_k(x, v)$ and a committed transaction $T_i$ writing $u$ to $x$ where $u \neq v$. Let $c_j$ be $r_k(x, v)$'s \lastw. Thus, $w_j(x, v) \in \evts{T_j}$. Thus, we have that $\ord{T_j} < \ord{T_k}$. Now there are two cases w.r.t $T_i$: (1) Suppose $\ord{T_i} < \ord{T_j}$. We now have that $T_i \ll T_j$. In this case, the mv edge is from $T_i$ to $T_j$. (2) Suppose $\ord{T_i} < \ord{T_j}$ which implies that $T_j \ll T_i$. Since $S$ is legal, we get that $\ord{T_k} < \ord{T_i}$. This case also implies that there is an edge from $\ord{T_k}$ to $\ord{T_i}$. Hence, in this case as well the \mv{} edges go from a transaction with lower \ordfn{} to a transaction with higher \ordfn{}. 
\end{itemize}

Thus, in all the three cases the edges go from a lower \ordfn{} transaction to higher \ordfn{} transaction. This implies that the the graph is acyclic. \qed
\end{proof}

\begin{lemma}
\label{lem:eqv_hist_mvorder}
Consider two histories $H, H'$ that are equivalent to each other. Consider a version order $\ll_H$ on the \tobj{s} created by $H$. The mv edges $\mv(H, \ll_H)$ induced by $\ll_H$ are the same in $H$ and $H'$.
\end{lemma}

\begin{proof}
Since the histories are equivalent to each other, the version order $\ll_H$ is applicable to both of them. It can be seen that the \mv{} edges depend only on events of the history and version order $\ll$. It does not depend on the ordering of the events in $H$. Hence, the \mv{} edges of $H$ and $H'$ are equivalent to each other. \qed
\end{proof}

\noindent Using these lemmas, we prove the following theorem. 

\begin{theorem}
\label{thm:opg}
A \valid{} history H is opaque iff there exists a version order $\ll_H$ such that $\opg{H}{\ll}$ is acyclic.
\end{theorem}

\begin{proof}
\textbf{(if part):} Here we have a a version order $\ll_H$ such that $G_H=\opg{H}{\ll}$ is acyclic. Now we have to show that $H$ is opaque. Since the $G_H$ is acyclic, a topological sort can be obtained on all the vertices of $G_H$. Using the topological sort, we can generate a \tseq{} history $S$. It can be seen that $S$ is equivalent to $\overline{H}$. Since $S$ is obtained by a topological sort on $G_H$ which maintains the real-time edges of $H$, it can be seen that $S$ respects the \rt{} order of $H$, i.e $\prec_H^{RT} \subseteq \prec_S^{RT}$. 

Similarly, since $G_H$ maintains reads-from order of $H$, it can be seen that if $T_j$ reads $x$ from $T_i$ in $H$ then $T_i$ terminates before $r_i(x)$ and $T_j$ in $S$. Thus, $S$ is \valid. Now it remains to be shown that $S$ is \legal. We prove this using contradiction. Assume that $S$ is not legal. Thus, there is a successful read \op{} $r_k(x, v)$ such that its \lastw{} in $S$ is $c_i$ and $T_i$ writes value $u (\neq v)$ to $x$, i.e $w_i(x, u) \in \evts{T_i}$. Further, we also have that there is a transaction $T_j$ that writes $v$ to $x$, i.e $w_j(x, v) \in \evts{T_j}$. Since $S$ is \valid, as shown above, we have that $T_j \prec_{S}^{RT} T_i \prec_{S}^{RT} T_k$.

Now in $\ll_H$, if $x_i \ll_H x_j$ then there is an edge from $T_i$ to $T_j$ in $G_H$. Otherwise ($x_j \ll_H x_i$), there is an edge from $T_k$ to $T_i$. Thus in either case $T_i$ can not be in between $T_j$ and $T_k$ in $S$ contradicting our assumption. This shows that $S$ is legal. 

\textbf{(Only if part):} Here we are given that $H$ is opaque and we have to show that there exists a version order $\ll$ such that $G_H=\opg{H}{\ll} (=\opg{\overline{H}}{\ll}$, \propref{hoverh}) is acyclic. Since $H$ is opaque there exists a \legal{} \tseq{} history $S$ equivalent to $\overline{H}$ such that it respects real-time order of $H$. Now, we define a version order for $S$, $\ll_S$ as in \defref{seqver}. Since the $S$ is equivalent to $\overline{H}$, $\ll_S$ is applicable to $\overline{H}$ as well. From \lemref{seracycle}, we get that $G_S=\opg{S}{\ll_S}$ is acyclic. Now consider $G_H = \opg{\overline{H}}{\ll_S}$. The vertices of $G_H$ are the same as $G_S$. Coming to the edges, 

\begin{itemize}
\item \rt{} edges: We have that $S$ respects real-time order of $H$, i.e $\prec_{H}^{RT} \subseteq \prec_{S}^{RT}$. Hence, all the \rt{} edges of $H$ are a subset of subset of $S$. 

\item \rf{} edges: Since $\overline{H}$ and $S$ are equivalent, the reads-from relation of $\overline{H}$ and $S$ are the same. Hence, the \rf{} edges are the same in $G_H$ and $G_S$. 

\item \mv{} edges: Since the version-order and the \op{s} of the $H$ and $S$ are the same, from \lemref{eqv_hist_mvorder} it can be seen that $\overline{H}$ and $S$ have the same \mv{} edges as well.
\end{itemize}

Thus, the graph $G_H$ is a subgraph of $G_S$. Since we already know that $G_S$ is acyclic from \lemref{seracycle}, we get that $G_H$ is also acyclic. \qed
\end{proof}

\section{Multiversion Timestamp Ordering (MVTO) Algorithm}
\label{sec:mvto}
We describe a timestamp based algorithm for multi-version STM systems, multiversion timestamp ordering (\mvto{}) algorithm. We then prove that our algorithm satisfies opacity \cite{tm-book, GuerKap:2008:PPoPP} using the graph characterization developed in the previous section. 

\subsection{The working principle} 
\label{subsec:princ} 
In our algorithm, each transaction, $T_i$ is assigned a unique timestamp, $i$, when it is initially invoked by a thread. We denote $i$ to be the id as well as the timestamp of the transaction $T_i$. Intuitively, the timestamp tells the ``time" at which the transaction began. It is a monotonically increasing number assigned to each transaction and is numerically greater than the timestamps of all the transactions invoked so far. All read and write operations carry the timestamp of the transaction that issued it. When an update transaction $T_i$ commits, the algorithm creates new version of all the \tobj{s} it writes to. All these versions have the timestamp $i$. 

Now we describe the main idea behind read, write and \tryc{} \op{s} executed by a transaction $T_i$. These ideas are based on the read and write steps for timestamp algorithm developed for databases by Bernstein and Goodman \cite{BernGood:1983:MCC:TDS}: 

\cmnt {
A STM system exports the functions:\begtrans{}, $\read_i, \writei_i$ and $\tryci_i$. A thread invokes a transaction with a \begtrans{} function which returns a unique transaction id which is the timestamp of the transactions. This timestamp is numerically greater than the timestamps of all the transactions invoked so far. This thread invokes future functions of the transaction using this timestamp. We use the notation $T_i$ to denote a transaction where $i$ is the timestamp of $T_i$. 
}

\begin{enumerate}
\item \textbf{read rule:} $T_i$ on invoking $r_i(x)$ reads the value $v$, where $v$ is the value written by a transaction $T_j$ that commits before $r_i(x)$ and $j$ is the largest timestamp $\leq i$.
\item \textbf{write rule:} $T_i$ writes into local memory.
\item \textbf{commit rule:} $T_i$ on invoking \tryc{} \op{} checks for each \tobj{} $x$, in its \Wset:
\begin{enumerate}
\item If a transaction $T_k$ has read $x$ from $T_j$, i.e. $r_k(x, v) \in \evts{T_k}$ and $w_j(x, v) \in \evts{T_j}$ and $j < i < k$, then $\tryc_i$ returns abort, 
\item otherwise, the transaction is allowed to commit.
\end{enumerate}
\end{enumerate}

\subsection{Data Structures and Pseudocode}
\label{subsec:code}

The algorithm maintains the following data structures. For each transaction $T_i$:

\begin{itemize}
\item $T_i.RS$(read set): It is a list of data tuples ($d\_tuples$) of the form $\langle x,v \rangle$, where $x$ is the t-object and $v$ is the value read from the transaction $T_i$.
\item $T_i.WS$(write set): It is a list of ($d\_tuples$) of the form $\langle x,v \rangle$, where $x$ is the \tobj{} to which transaction $T_i$ writes the value $v$.
\end{itemize}

For each transaction object ($t\_object$) $x$:
\begin{itemize}
\item $x.vl$(version list): It is a list consisting of version tuples ($v\_tuple$) of the form $\langle ts,v,rl \rangle$ where $ts$ is the timestamp of a committed transaction that writes the value $v$ to $x$. The list $rl$ is the read list consisting of a set of transactions that have read the value $v$ (described below). Informally the version list consists of all the committed transaction that have ever written to this t-object and the set of corresponding transactions that have read the value $v$ on $x$.

\item $rl$(read list): This list contains all the read transaction tuples ($rt\_tuples$) of the form $\langle j \rangle$. The read list $rl$ is stored in each tuple of the version list described above.
\end{itemize}
Figure 3 illustrates the how the version list and read list are managed.

\begin{figure}
\label{fig:ds}
\centering
\includegraphics[scale=.5]{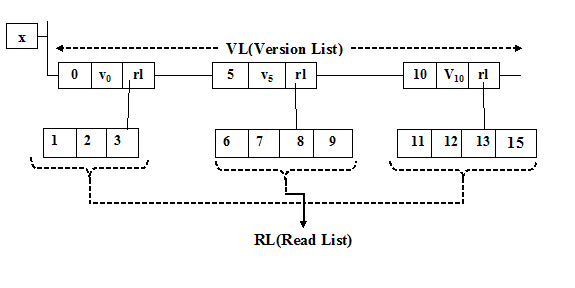}
\caption{Data Structures}
\end{figure}

In addition, the algorithm maintains two global data-structures:
\begin{itemize}
\item \tcntr: This counter is used to generate the ids/timestamp for a newly invoked transaction. This is incremented everytime a new transaction is invoked.
\item \livel: This list keeps track of all the transactions that are currently incomplete or live. When a transaction begins, its id is added to this list. When it terminates (by abort or commit), the id is deleted from this list. 
\end{itemize}

%
The STM system consists of the following operations/functions. These are executed whenever a transaction begins, reads, write or tries to commit:
 
\vspace{1mm}
\noindent
\textit{$\init():$} This operation initializes the STM system. It is assumed that the STM system knows all the \tobj{s} ever accessed. All these \tobj{s} are initialized with value 0 by the initial transaction $T_0$ in this operation. A version tuple $\langle 0, 0, nil \rangle$ is inserted into all the version list of all the \tobj{s}. 

\vspace{1mm}
\noindent
\textit{$\begtrans():$} A thread invokes a transaction by executing this operation. It returns an unique transaction identifier which is also its timestamp. The id is used in all other operations exported by the STM system. The id is further stored in the $\livel$.

\vspace{1mm}
\noindent
\textit{$read_i(x):$} To read any \tobj{} by transaction $i$, this operation is invoked. First, the \tobj{} $x$ is locked. Then the version list of $x$ is searched to identify the correct $version\_tuple$ (i.e the version created by a writing transaction). From the version-list, the tuple with the largest timestamp  less than $i$, say $\langle j, v \rangle$ is identified. Then the value $v$ written by transaction $j$, is returned. 

\vspace{1mm}
\noindent
\textit{$\find(i,x):$} This function is invoked by $read_i(x)$ and finds the tuple $\langle j,v,rl \rangle$ having the largest timestamp $j$ value smaller than $i$ (lts). 

\vspace{1mm}
\noindent
\textit{$write_i(x,v):$} Here write is performed in the local memory. This operation appends the data tuple $\langle x, v \rangle$ into the WS of transaction $T_i$.

\vspace{1mm}
\noindent
\textit{$\tryc_i():$} This operation is invoked when a transaction $T_i$ has completed all its \op{s} and wants to commit. This operation first checks whether $T_i$ is read only or not. If it is read only transaction then it returns commit. Otherwise, for each \tobj{} $x$ (accessed in a predefined order) in $T_i$'s write set, the following check is performed: if timestamp of $T_i$, $i$ between the timestamps of the $T_j$ and $T_k$, where transaction $T_k$ reads $x$ from transaction $T_j$, i.e $j < i < k$, then the transaction $T_i$ is aborted. 

If this check succeeds for all the \tobj{s} written by $T_i$, then the version tuples are appended to the version lists and the transaction $T_i$ is committed. Before returning either commit or abort, the transaction id $i$ is removed from \livel.

The system orders all the \tobj{s} ever accessed as $x_1, x_2, ...., x_n$ by any transaction (assuming that the system accesses a total of $n$ \tobj{s}). In this \op, each transaction locks and access \tobj{s} in the increasing order which ensures that the system does not deadlock.

\vspace{1mm}
\noindent
\textit{$\checkv(i,x):$} This function checks the version list of $x$. For all version tuples $\langle j, v, rl \rangle$ in $x.vl$ and for all transactions $T_k$ in $rl$, it checks if the timestamp of $T_i$ is between the timestamp of the $T_j$ and $T_k$, i.e $j < i < k$. If so, it returns true else false.

%
\begin{algorithm} 
\label{alg:init} 
\caption{STM $\init()$: Invoked at the start of the STM system. Initializes all the \tobj{s} used by the STM System}
\begin{algorithmic}[1]
\ForAll {$x$ used by the STM System}
\State /* $T_0$ is initializing $x$ */
\State add $\langle 0, 0, nil \rangle$ to $x.vl$; \label{lin:t0-init} 
\EndFor;
\end{algorithmic}
\end{algorithm}

\begin{algorithm} 
\label{alg:begin} 
\caption{STM $\begtrans()$: Invoked by a thread to being a new transaction $T_i$}
\begin{algorithmic}[1]
\State lock $\livel$;
\State \Comment{Store the latest value of $\tcntr$ in $i$.} 
\State $i = \tcntr$; 
\State $\tcntr = \tcntr + 1$; 
\State add $i$ to $\livel$; 
\State unlock $\livel$;
\State return $i$;
\end{algorithmic}
\end{algorithm}

\begin{algorithm}  
\label{alg:read} 
\caption{STM $read_i(x)$: A Transaction $T_i$ reads t-object $x$}
\begin{algorithmic}[1]
\State lock $x$;
\label{lin:rlock}
\State \Comment{From $x.vls$, identify the right $version\_tuple$.} 
\State $\langle j,v,rl \rangle  = \find(i,x)$;	
\State Append $i$ into $rl$; \Comment{Add $i$ into $j$'s $rl$.}
\State \emph{unlock $x$};
\State return $(v)$; \Comment{v is the value returned}
\end{algorithmic}
\end{algorithm}

\begin{algorithm}  
\label{alg:select}
\caption{$\find(i,x)$: Finds the tuple $\langle j,v,rl \rangle$ created by the transaction $T_j$ with the largest timestamp smaller than $i$}
\begin{algorithmic}[1]
\State \Comment{Initialize $closest\_tuple$}
\State $closest\_tuple = \langle 0,0,nil \rangle$;
\ForAll {$\langle k,v,rl \rangle \in x.vl$}
\If {$(k < i)$ and $(closest\_tuple.ts < k)$}
\State $closest\_tuple = \langle k,v, rl \rangle$; 
\EndIf;
\EndFor;
\State return $(closest\_tuple)$;
\end{algorithmic}
\end{algorithm}

\begin{algorithm}  
\label{alg:write} 
\caption{STM $write_i(x,v)$: A Transaction $T_i$ writes into local memory}
\begin{algorithmic}[1]
\State Append the $d\_tuple \langle x,v \rangle$ to $T_i.WS$.
\State return $ok$;
\end{algorithmic}
\end{algorithm}

\begin{algorithm}  
\label{alg:tryc} 
\caption{STM $\tryc()$: Returns $ok$ on commit else return Abort}
\begin{algorithmic}[1]
\If {$(T_i.WS == NULL)$}
\State $\remid(i)$;
\State return $ok$; \Comment{A read-only transaction.}
\EndIf;	
\ForAll {$d\_tuple (x,v)$ in $T_i \cdot WS$}
\State /* Lock the \tobj{s} in a predefined order to avoid deadlocks */
\State Lock $x$;	
\If {$(\checkv(i,x) == false)$}
\State unlock all the variables locked so far;
\State $\remid(i)$;
\State return $Abort$;
\EndIf;
\EndFor;
\State /* Successfully checked for all the write variables and not yet aborted. So the new write versions can be inserted. */
\ForAll {$d\_tuples \langle x,v \rangle $ in $T_i.WS$}
\State insert $v\_tuple \langle i,v,nil \rangle$ into $x.vl$ in the increasing order;  \label{lin:ins-vtuple}
\EndFor;
\State unlock all the variables;
\State $\remid(i)$;
\State return $ok$;
\end{algorithmic}
\end{algorithm}

\begin{algorithm}  
\label{alg:checkVersion} 
\caption{$\checkv(i,x)$:Checks the version list; it returns True or false}
\begin{algorithmic}[1]
\ForAll {$v\_tuples \langle j,v,rl \rangle $ in $x \cdot vl$}
\ForAll {$T_k$ in $rl$}
\State /* $T_k$ has already read the version created by $T_j$ */
\If {$(j < i < k)$}   
\State return $false$;  
\EndIf;
\EndFor;
\EndFor;
\State return $true$;
\end{algorithmic}
\end{algorithm}

\begin{algorithm}  
\label{alg:removeid} 
\caption{$\remid(i)$:Removes transaction id $i$ from the $\livel$}
\begin{algorithmic}[1]
\State lock $\livel$; 
\State remove $i$ from $\livel$; 
\State unlock $\livel$; 
\end{algorithmic}
\end{algorithm}

\subsection{Proof of MVTO protocol}
\label{subsec:proof}

In this sub-section, we will prove that our implementation satisfies opacity. Consider the history $H$ generated by \mvto{} algorithm. Recall that only the \begtrans{}, read, and \tryc{} \op{s} access shared memory. Hence, we call such operations \emph{\memop{s}}. 

Note that $H$ is not necessarily sequential: the transactional \op{s} can execute in overlapping manner. To reason about correctness we have to prove $H$ is opaque. Since we defined opacity for histories which are sequential, we order all the overlapping \op{s} in $H$ to get an equivalent sequential history. We then show that this resulting sequential history satisfies \op.

We order overlapping \memop{s} of $H$ as follows: (1) two overlapping \begtrans{} \op{s} based on the order in which they obtain lock over \tcntr; (2) two read \op{s} accessing the same \tobj{} $x$ by their order of obtaining lock over $x$; (3) a read $r_i(x)$ and a $\tryc_j$, of a transaction $T_j$ which has written to $x$, are similarly ordered by their order of obtaining lock over $x$; (4) \begtrans{} and a \tryc{} \op{s} are ordered by their order of obtaining locks over $\livel$; (5) similarly, two \tryc{} \op{s} based on the order in which they obtain lock over $\livel$.

Combining the real-time order of events with above mentioned order, we obtain a partial order which we denote as \emph{$\locko_H$}. (It is a partial order since it does not order overlapping read \op{s} on different \tobj{s} or an overlapping read and a \tryc{} which do not access any common \tobj). 

In order for $H$ to be to sequential, all its \op{s} must be ordered. Let $\alpha$ be a total order or \emph{linearization} of \op{s} of $H$ such that when this order is applied to $H$, it is sequential. We denote the resulting history as $H^{\alpha} = \seq{\alpha}{H}$. 
We now argue about the \validity{} of histories generated by the algorithm. 

\begin{lemma}
\label{lem:histvalid}
Consider a history $H$ generated by the algorithm. Let $\alpha$ be a linearization of $H$ which respects $\locko_H$, i.e. $\locko_H \subseteq \alpha$. Then $H^{\alpha} = \seq{\alpha}{H}$ is \valid. 
\end{lemma}

\begin{proof}
Consider a successful read \op{} $r_i(x)$ that returns value $v$. The read function first obtains lock on \tobj{} $x$ (Algorithm $read$, Line \ref{lin:rlock}). Thus the value $v$ returned by the read function must have already been stored in $x$'s version list by a transaction, say $T_j$ when it successfully returned ok from its \tryc{} \op{} (if $T_j \neq T_0$). For this to have occurred, $T_j$ must have successfully locked and released $x$ prior to $T_i$'s locking \op. Thus from the definition of $\locko_H$, we get that $\tryc_j(ok)$ occurs before $r_i(x,v)$ which also holds in $\alpha$. 

If $T_j$ is $T_0$, then by our assumption we have that $T_j$ committed before the start of any \op{} in $H$. Hence, this automatically implies that in both cases $H^{\alpha}$ is \valid. 
\end{proof}

It can be seen that for proving correctness, any linearization of a history $H$ is sufficient as long as the linearization respects $\locko_H$. The following lemma formalizes this intuition, 
\begin{lemma}
\label{lem:histseq}
Consider a history $H$. Let $\alpha$ and $\beta$ be two linearizations of $H$ such that both of them respect $\locko_H$, i.e. $\locko_H \subseteq \alpha$ and $\locko_H \subseteq \beta$. Then, (1) $H^{\alpha} = \seq{\alpha}{H}$ is opaque iff $H^{\beta} = \seq{\beta}{H}$. 
\end{lemma}

\begin{proof} From \lemref{histvalid}, we get that both $H^{\alpha}$ and $H^{\beta}$ are \valid{} histories. Now let us consider each case \\
\textbf{If:}  Assume that $H^{\alpha}$ is opaque. Then, we get that there exists a legal \tseq{} history $S$ that is equivalent to $\overline{H^{\alpha}}$. From the definition of $H^{\beta}$, we get that $\overline{H^{\alpha}}$ is equivalent to $\overline{H^{\beta}}$. Hence, $S$ is equivalent to $\overline{H^{\beta}}$ as well. We also have that, $\prec_{H^{\alpha}}^{RT} \subseteq \prec_{S}^{RT}$. From the definition of $\locko_H$, we get that $\prec_{H^{\alpha}}^{RT} = \prec_{\locko_H}^{RT} = \prec_{H^{\beta}}^{RT}$. This automatically implies that $\prec_{H^{\beta}}^{RT} \subseteq \prec_{S}^{RT}$. Thus $H^{\beta}$ is opaque as well. 

~ \\
\textbf{Only if:} This proof comes from symmetry since $H^{\alpha}$ and $H^{\beta}$ are not distinguishable. 
\end{proof}

\cmnt{
We consider the \memop{s} one by one. We order two overlapping \begtrans{} \op{s} by based on the order in which they obtain lock over \tcntr. A \begtrans{} \op{} and any other \memop{} can be ordered in any way compatible with the rest of the \op{s}. A read \op{} $r_i(x, v)$, and the \tryc{} \op{} of a transaction $T_j$ writing to $x$, $w_j(x, u)$ are ordered by their access on $x$. Similarly, two \tryc{} \op{s} of transaction $T_i$, $T_j$ accessing common \tobj{s} are ordered by their access to the common \tobj{s}. It must be noted that the transactions access all the \tobj{s} in a given order. Hence, deadlocks are not possible. 

Using this ordering, we get the following lemma,

\begin{lemma}
\label{lem:histseq}
The overlapping \op{s} of the history generated by \mvto{} algorithm can be ordered so that it is sequential. 
\end{lemma}
}

This lemma shows that, given a history $H$, it is enough to consider one sequential history $H^{\alpha}$ that respects $\locko_H$ for proving correctness. If this history is opaque, then any other sequential history that respects $\locko_H$ is also opaque.

\cmnt {
\vspace{1mm}
\noindent
\textit{Notations:}
}
Consider a history $H$ generated by \mvto{} algorithm. We then generate a sequential history that respects $\locko_H$. For simplicity, we denote the resulting sequential history as $H_{to}$. Let $T_i$ be a committed transaction in $H_{to}$ that writes to $x$ (i.e. it creates a new version of $x$). 

To prove the correctness, we now introduce some more notations. We define $\stl{T_i}{x}{H_{to}}$ as a committed transaction $T_j$ such that $T_j$ has the smallest timestamp greater than $T_i$ in $H_{to}$ that writes to $x$ in $H_{to}$. Similarly, we define $\lts{T_i}{x}{H_{to}}$ as a committed transaction $T_k$ such that $T_k$ has the largest timestamp smaller than $T_i$ that writes to $x$ in $H_{to}$. \cmnt{ We denote $\vli{ts}{x}{H_{to}}$, as the $v\_tuple$ in $x.vl$ after executing all the events in $H_{to}$ created by transaction $T_{ts}$. If no such $v\_tuple$ exists then it is nil. }Using these notations, we describe the following properties and lemmas on $H_{to}$,

\begin{property}
\label{prop:uniq-ts}
Every transaction $T_i$ is assigned an unique numeric timestamp $i$.
\end{property}

\begin{property}
\label{prop:ts-inc}
If a transaction $T_i$ begins after another transaction $T_j$ then $j < i$.
\end{property}

\begin{property}
\label{prop:readsfrom}
If a transaction $T_k$ reads $x$ from (a committed transaction) $T_j$ then $T_j$ is a committed transaction with $j$ being the largest timestamp smaller than $k$. Formally, $T_j = \lts{x}{T_k}{H_{to}}$. 
\end{property}

\cmnt{
\begin{proof}
This comes directly from the execution of read function. 
\end{proof}

\begin{property}
\label{prop:readswrite}
Suppose a transaction $T_k$ reads $x$ from (a committed transaction) $T_j$. Let $T_i$ be a committed transaction that writes to $x$, i.e. $w_i(x, u) \in \evts{T_i}$. Then, the timestamp of $T_i$ is either less than $T_j$'s timestamp or greater than $T_k$'s timestamp, i.e. $i<j \vee k>i$. 
\end{property}
}

\begin{lemma}
\label{lem:readswrite}
Suppose a transaction $T_k$ reads $x$ from (a committed transaction) $T_j$ in $H_{to}$, i.e. $\{w_j(x, v), r_k(x, v)\} \in \evts{H_{to}}$. Let $T_i$ be a committed transaction that writes to $x$, i.e. $w_i(x, u) \in \evts{T_i}$. Then, the timestamp of $T_i$ is either less than $T_j$'s timestamp or greater than $T_k$'s timestamp, i.e. $i<j \oplus k<i$ (where $\oplus$ is xor operator). 
\end{lemma}

\begin{proof}
We will prove this by contradiction. Assume that $i<j \oplus k<i$ is not true. This implies that, $j<i<k$. But from the implementation of read and \tryc{} functions, we get that either transaction $T_i$ is aborted or $T_k$ reads $x$ from $T_i$ in $H$. Since neither of them are true, we get that $j<i<k$ is not possible. Hence, $i<j \oplus k<i$. 
\end{proof}

To show that $H_{to}$ satisfies opacity, we use the graph characterization developed in \secref{gckt}. For the graph characterization, we use the version order defined using timestamps. Consider two committed transactions $T_i, T_j$ such that $i < j$. Suppose both the transactions write to \tobj{} $x$. Then the versions created are ordered as: $x_i \ll x_j$. We denote this version order on all the \tobj{s} created as $\ll_{to}$. Now consider the opacity graph of $H_{to}$ with version order as defined by $\ll_{to}$, $G_{to} = \opg{H_{to}}{\ll_{to}}$. In the following lemmas, we will prove that $G_{to}$ is acyclic.

\begin{lemma}
\label{lem:edgeorder}
All the edges in $G_{to} = \opg{H_{to}}{\ll_{to}}$ are in timestamp order, i.e. if there is an edge from $T_j$ to $T_i$ then the $j < i$. 
\end{lemma}

\begin{proof}
To prove this, let us analyze the edges one by one, 
\begin{itemize}
\item \rt{} edges: If there is a \rt{} edge from $T_j$ to $T_i$, then $T_i$ terminated before $T_j$ started. This implies that $T_j$ started before $T_i$. Hence, from \propref{ts-inc} we get that $j < i$.

\item \rf{} edges: This follows directly from \propref{readsfrom}. 

\item \mv{} edges: The \mv{} edges relate a committed transaction $T_i$ writing to a \tobj{} $x$, $w_i(x,u)$; a successful read \op{} $r_k(x,v)$ belonging to a transaction $T_k$ reading $x$ written by a committed transaction $T_j$, $w_j(x, v)$ and transaction $T_j$. Transactions $T_j, T_i$ create new versions $x_i, x_j$ respectively. According to $\ll_{to}$, if $x_i \ll_{to} x_j$, then there is an edge from $T_i$ to $T_j$. From the definition of $\ll_{to}$ this automatically implies that $i < j$.

On the other hand, if $x_j \ll_{to} x_i$ then there is an edge from $T_k$ to $T_i$. Thus in this case, we get that $j < i$. Combining this with \lemref{readswrite}, we get that $k < i$. 
\end{itemize}
Thus in all the cases we have shown that if there is an edge from $T_j$ to $T_i$ then the $j < i$. 
\end{proof}

\begin{theorem}
\label{thm:to-opaque}
The history $H_{to}$ is opaque. 
\end{theorem}

\begin{proof}
From the definition of $H_{to}$ and \lemref{histvalid}, we get that $H_{to}$ is \valid. We show that $G_{to} = \opg{H_{to}}{\ll_{to}}$ is acyclic. We prove this by contradiction. Assume that $G_{to}$ contains a cycle of the form, $T_{c1} \rightarrow T_{c2} \rightarrow .. T_{cm} \rightarrow T_{c1}$. From \lemref{edgeorder} we get that, $c1 < c2 < ... < cm < c1$ which implies that $c1 < c1$. Hence, a contradiction. This implies that $G_{to}$ is acyclic. Thus from \thmref{opg} we get that $H_{to}$ is opaque. 
\end{proof}

Now, it is left to show that our algorithm is \emph{live}, i.e., under certain conditions, every operation eventually completes. We have to show that the transactions do not deadlock. The is so because all the transactions lock all the \tobj{s} in a predefined order. As discussed earlier, the STM system the orders all \tobj{s}. We denote this order as \aco and denote it as $\prec_{ao}$. Thus $x_1 \prec_{ao} x_2 \prec_{ao} ... \prec_{ao} x_n$. In addition to \tobj{s}, the transactions also access the shared variable $\livel$. Thus we add $\livel$ to this order: $x_n \prec_{ao} \livel$. We refer to the combined set of \tobj{s} $x_1, .... x_n$ and $\livel$ as shared objects. 

From \aco, we get the following property

\begin{property}
\label{prop:aco}
Suppose transaction $T_i$ accesses shared objects $p$ and $q$ in $H$. Iff $p$ is ordered before $q$ in \aco, then $lock(p)$ by transanction $T_i$ occurs before $lock(q)$. Formally, $(p \prec_{ao} q) \Leftrightarrow (lock(p) <_H lock(q))$. 
\end{property}

\begin{theorem}
\label{thm:live}
Assuming that no transaction fails and all the locks are starvation-free, every operation of \mvto{} algorithm eventually returns. 
\end{theorem}

\begin{proof}
In our algorithm, a transaction $T_k$ executing some operation will not return only if the \op{} or a sub-function that is invoked by the \op{} is stuck waiting on a lock. This is possible only when a set of transactions, denoted as $D$ (which includes $T_k$) are deadlocked. Let $SO$ be the set of all shared objects locked by transactions in the $D$. Let $s$ be a shared object in the set $SO$ that is ranked highest according to \aco{} and locked by a transaction in $T_i$ in $D$. Since $T_i$ is deadlocked, it must be waiting to access a shared object, say $s' \in SO$ locked by a transaction $T_j$ (otherwise $T_i$ cannot be involved in the deadlocked). From \propref{aco}, we get that $s \prec_{ao} s'$. But this contradicts our choice of $s$. Hence, a deadlock is not possible. 
\end{proof}

Finally, we prove that  our algorithm satisfies mv-permissive.

\begin{theorem}
\label{thm:mvp}
\mvto{} algorithm is mv-permissive. 
\end{theorem}

\begin{proof}
According to the algorithm a read \op{} never returns abort. Hence, when a read-only transaction executes \tryc{} \op, it always $ok$. Thus, a read-only transaction never aborts.Further, an update transaction aborts only if another update transaction has already committed a previous version. This shows that \mvto{} algorithm is mv-permissive.
\end{proof}

\section{Garbage Collection}
\label{sec:gar}
As one can see with multi-version STMs, multiple versions are created. But storing multiple versions can unnecessarily waste memory. Hence, it is important to perform garbage collection by deleting unwanted versions of \tobj{s}. Some of the earlier STM systems solve this problem by maintaining a fixed number of versions for each \tobj~\cite{CachRito:2006:CP}. We on the other hand, do not restrict the number of versions. The STM system will detect versions that will never again be used (i.e. have become garbage) and delete them. The garbage collection routine will be invoked whenever the number of versions of a \tobj{} has become greater than a predefined threshold. The threshold can be decided dynamically by the application invoking the STM system based on the current memory requirements. 

The STM system will delete a version of a \tobj{} $x$ created by transaction $T_i$ when the following conditions are satisfied:

\begin{enumerate}
\item At least one another version of $x$ has been created by $T_k$ and $i < k$;
\item Any transaction $T_j$ such that $T_i < T_j < T_k$ has terminated (either committed or aborted).
\end{enumerate}

To capture these conditions, we modify the data structure maintained. For each \tobj{} $x$, we augment the version tuples stored in its version list by adding another entry $nts$. Thus each $v\_tuple$ is: $\langle ts,v,rl,nts \rangle$. The entry $nts$ (next timestamp) denotes a committed transaction with the smallest timestamp larger than $ts$ that has created a version of $x$. If there is no such transaction then $nts$ is $nil$. With this modification, we make the following changes to functions discussed in the previous section:

\vspace{1mm}
\noindent
\textit{$\init()$:} In the version tuple created for $T_0$ on every \tobj{} $x$, the $nts$ entry is also initialized to $nil$. Thus, the \lineref{t0-init} is replaced with: add $\langle 0, 0, nil, nil \rangle$ to $x.vl$;

\vspace{1mm}
\noindent
\textit{$read_i(x), \find(i,x):$} The changes in both these functions are trivial. Whenever the version tuple is referred to, the $nts$ entry is also assumed to be present. There field is not directly used in these functions. 

\vspace{1mm}
\noindent
\textit{$\tryc_i():$} We replace \lineref{ins-vtuple} of \tryc{} with the function $\instup(x, i, v, nil)$. The description of this function is given below.

\vspace{1mm}
\noindent
\textit{$\instup(x, i, v, nil):$} This function inserts the new value written by transaction $i$ into $x$'s version list. First it creates a version tuple, $cur\_tuple$, with timestamp $i$ and value $v$. It then identifies the previous version tuple created by a transaction that has the largest id (timestamp) smaller than $i$ (\lineref{prevtup}). Then it copies the $nts$ entry of the $prev\_tuple$ into $cur\_tuple$. The $nts$ entry of $prev\_tuple$ is stored as $i$. This way, the $nts$ entry for every version tuple is kept updated. Then garbage collection routine is invoked if the total number of $v\_tuples$ of $x$ is greater than a predefined threshold. 

Having described the changes necessary to keep $nts$ entry updated, we will next describe the steps to perform garbage collection, $\gc()$. 

\vspace{1mm}
\noindent
\textit{$\gc(x):$} On being invoked, this function locks $\livel$. It checks the version lists of the currently considered \tobj{} $x$. For each tuple in $x$'s version list, denoted as $cur\_tuple$, first $cur\_tuple.nts$ is checked. If it $nil$, then the next tuple is checked. Otherwise, for each transaction with timestamp $j$ in the range $cur\_tuple.ts$ to $cur\_tuple.nts$ is checked. If some transaction $T_j$ is in $\livel$, then the algorithm decides that $cur\_tuple$ is not yet garbage. It then checks the next tuple is checked. If no such transaction $T_j$ in this range is in $\livel$, then the algorithm decides that $cur\_tuple$ has become garbage. Hence, it deletes this tuple. Then, it returns the control. But before returning this function does not unlock $\livel$ as $\livel$ will again be locked in \tryc. 


Figure 4 illustrates the idea of garbage collection. Here, the version tuple created by transaction $T_5$ has already been deleted. The version tuple created by transaction $T_{10}$ will be deleted when all the transactions between 10 and 17 have terminated (either aborted or committed).

\begin{algorithm}  
\label{alg:instuple}
\caption{$\instup(x, i, v, nil)$: Inserts the version tuple for $(i,v)$ created by the transaction $T_i$ into the version list of $x$}
\begin{algorithmic}[1]
\State \Comment{Initialize $cur\_tuple$}
\State $cur\_tuple = \langle i,v,nil,nil \rangle$; 
\State /* Finds the tuple with the largest timestamp smaller than i */
\State $prev\_tuple = \find(i, x)$; \label{lin:prevtup}
\State \Comment{$prev\_tuple$ is of the form $\langle ts, v', rl, nts\rangle$}
\State $cur\_tuple.nts = prev\_tuple.nts$;
\State $prev\_tuple.nts = i$;
\State insert $cur\_tuple$ into $x.vl$ in the increasing order of timestamps;  
\State /* $\lvert x.vl \lvert$ denotes number of versions of $x$ created and threshold is a predefined value. */
\If {($\lvert x.vl \lvert > threshold$)} 
\State $gc(x)$;
\EndIf 
\end{algorithmic}
\end{algorithm}

\begin{algorithm}
\caption{STM $\gc(x)$: Unused version of a \tobj{} $x$ will deleted from $x.vl$}
\begin{algorithmic}[1]
\State lock $\livel$;
\State \Comment{\tobj{} $x$ is already locked}
\ForAll {$(cur\_tuple \in x.vl)$} 
\If {$(cur\_tuple.nts == nil)$}
\State /* If $nts$ is nil, check the next tuple in the version list */
\State continue; 
\EndIf
\State $j = cur\_tuple.ts + 1$;
\State \Comment{Check for all ids $j$ in the range $j < nts$}
\While {$(j < cur\_tuple.nts)$} 
\If {$(j \in \livel)$}
\State break;
\EndIf
\EndWhile
\State /* If all the tuples with timestamp $j$, such that  $i< j < nts$ have terminated then $cur\_tuple$ can be deleted*/
\State delete $cur\_tuple$;
\EndFor 
\State /* $\livel$ is not unlocked when this function returns */
\end{algorithmic}
\end{algorithm}
\cmnt {
\vspace{1mm}
\noindent
\textit{$garb\_collect()$:} In the version list, each tuple is in the form of $(i,v,r_l)$. So if we want to delete the version $i$, we have to check it whether its next version has been created or not according to the MVTO protocol. If the next version has not been created then the version $i$ can not be deleted.

Suppose the next version of $(i,v,r_l)$ is $(i',v',r_l')$ and we want to delete $(i,v,r_l)$ version. According to this protocol, it will check the TS of all the $j$ transactions between $i$ and $i'$. If $(j \ni r_l) \wedge (j \in live\_trans\_list()$, then it will go to the next $j+1$ transaction and repeat the same process. When the process complete, the STM system automatically delete the version $(i,v,r_l)$.

In $\tryc$ function, before line no 2, 8 and 15, the id of the transaction from $live\_trans\_list()$ will be removed. In line no 15 garbage collection function is invoked.
}

\begin{figure}
\label{fig:gcds}
\centering
\includegraphics[scale=0.5]{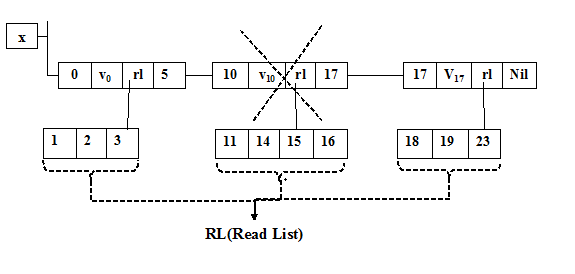}
\caption{Data Structures for Garbage Collection}
\label{fig:3}
\end{figure}

\subsection{Proof of Garbage Collection}
\label{subsec:gcproof}

Consider a history $H$ generated by the MVTO algorithm with garbage collection. As discussed \secref{mvto}, $H$ is not sequential. To prove the correctness, we order the overlapping \op{s} to obtain a sequential history. Similar to \secref{mvto} we use a total order that respects $\locko_H$ to order the overlapping \op{s}. Although the \tryc{} function is modified due to invocation of garbage collection functions, \instup{} and \gc, it does not modify the \locko. 

Thus \lemref{histvalid} and \lemref{histseq} is applicable to $H$. Hence, we consider any total order that respects $\locko_H$ for ordering the overlapping \op{s} of $H$. We denote the resulting sequential history as $H_{gc}$.

To prove of our garbage collection scheme, we now introduce some more notations. \cmnt{Let $T_i$ be a committed transaction in $H_{gc}$ that writes to $x$ (i.e. it creates a new version of $x$). Then, we define $\stl{x}{T_i}{H_{gc}}$ as a committed transaction $T_j$ such that $T_j$ has the smallest timestamp greater than $T_i$ in $H_{gc}$ that writes to $x$ in $H_{gc}$. Similarly, we define $\lts{x}{T_i}{H_{gc}}$ as a committed transaction $T_k$ such that $T_k$ has the largest timestamp smaller than $T_i$ that writes to $x$ in $H_{gc}$. }We denote $\vli{ts}{x}{H_{gc}}$, as the $v\_tuple$ in $x.vl$ created by transaction $T_{ts}$ in $H_{gc}$. If no such $v\_tuple$ exists then it is nil. We have the following useful lemmas on garbage collection. In these lemmas, we use the notations defined in \subsecref{proof}.

\begin{lemma}
\label{lem:vlreal}
Consider any history $H_{gc}$ generated by the algorithm with garbage collection. Let $H_p$ be a prefix of $H_{gc}$. For every live transaction $T_l$ in $incomplete(H_p)$ and for every \tobj{} $x$, we have: (a) Let $T_j = \lts{T_l}{x}{H_p}$. Then $\vli{T_j}{x}{H_p} = vt_j \neq nil$. (b) Let $vt_j.nts = T_k$. Then, $T_k = \stl{T_l}{x}{H_p}$
\end{lemma}

\begin{proof}
We prove this using induction on number of version tuples, $count$, created for \tobj{} $x$. 

\vspace{1mm}
\noindent
\textit{Base case, $count=0$}: When the STM system is initialized, the first version is created by $T_0$. Consider a prefix of $H_{gc}$, denoted as $H_0$, which has only this version of $x$ created by $T_0$. Thus in $H_0$ no transaction has yet committed. Since no transaction has yet executed \tryc{} in $H_0$, the \gc{} function would not have been executed. So the version tuple created by $T_0$ would not have been deleted. Let $T_l$ be a live transaction in $H_0$. So we get that $T_0 = \lts{T_l}{x}{H_0}$, $\vli{T_0}{x}{H_0} = vt_0 \neq nil$. We also have that $vt_0.nts = nil$ and $\stl{T_l}{x}{H_{gc}} = nil$. So this proves the base case. 

\vspace{1mm}
\noindent
\textit{Induction case, $count=m+1$}: In this case, we have to prove the lemma after $m+1$ version tuples have been created assuming that it is true when $m$ tuples were present. Consider a prefix of $H_{gc}$, denoted as $H_m$, in which $m$ versions are created (note that the number of versions of $x$ present in $H_m$ could be less than $m$ since some tuples could have been deleted by \gc() function). Let $T_l$ be a live transaction in $H_m$ that executes \tryc{} to generate the next version of $x$. From induction hypothesis, we have that $T_j = \lts{T_l}{x}{H_m}$, $\vli{T_j}{x}{H_m} = vt_j \neq nil$ and $vt_j.nts = T_k = \stl{T_l}{x}{H_m}$.

Consider another prefix of $H_{gc}$, $H_{m+1}$ in which $T_l$ committed and created $m+1$st version of $x$. As observed earlier, the lemma is true in $H_m$. When $T_l$ commits, the only live transactions that are affected are those transactions whose timestamps are between $j$ and $k$. Thus if we prove that it is true for all these live transactions, then the lemma is true for all live transactions in $H_{m+1}$. Consider two live transactions $T_{l1}$ and $T_{l2}$ such that $j<l1<l<l2<k$. From the \tryc{} \op{} of $T_l$, we get that $T_j = \lts{T_{l1}}{x}{H_m} = \lts{T_{l1}}{x}{H_{m+1}}$ and $\vli{T_j}{x}{H_m} = \vli{T_j}{x}{H_{m+1}} = vt_j \neq nil$. We also have that $T_k = \stl{T_{l2}}{x}{H_m} = \stl{T_{l2}}{x}{H_{m+1}}$. 

Since $T_l$ is committed in $T_{m+1}$ a new tuple of $x$ is created. Thus we get that $\lts{T_{l2}}{x}{H_{m+1}} = T_l$ and $ \vli{T_l}{x}{H_{m+1}} = vt_l \neq nil$. This proves the induction case. 
\end{proof}

This lemma intuitively states that for any live transaction $T_l$, its \ltsn{} transaction for \tobj{} $x$, $T_j$ is not deleted by \gc{} function. It also states that for all version tuples, the $nts$ entry is correctly maintained. Using this lemma, we next prove that that \propref{readsfrom} and \lemref{readswrite} are true even with garbage collection.

\begin{lemma}
\label{lem:gc-read-rule}
The history $H_{gc}$ generated by MVTO with garbage collection satisfies read rule: If a transaction $T_k$ reads $x$ from (a committed transaction) $T_j$ in $H_{gc}$ then $T_j$ is a committed transaction with $j$ being the largest timestamp smaller than $k$. Formally, $T_j = \lts{T_k}{x}{H_{gc}}$. 
\end{lemma}

\begin{proof}
Consider a history $H'$ (a prefix of $H_{gc}$) in which the read \op{} of $T_k$ is the last \op{} to execute in $H'$. Thus $T_k$ is a live transaction in $H'$. Let $T_j = \lts{T_k}{x}{H'}$. From \lemref{vlreal}.(a), we get that $T_k$'s \ltsn{} transaction is correctly maintained by the algorithm (with garbage collection). Hence, $\vli{T_j}{x}{H'} = vt_j \neq nil$ (which implies that the tuple $vt_j$ has not yet been deleted). Hence, $T_k$ reads $x$ from $T_j$.

\cmnt {
Now suppose $T_j \neq T_i = \lts{T_k}{x}{H_{gc}}$. This is possible only if $j<i<k$. 

We assume that $T_i$ is the first transaction to commit so that it comes between $j$ and $k$. Consider the prefix of $H_{gc}$, $H''$ which has all the events until $T_i$'s commit event but not the commit event. Hence, $T_i$ is a live transaction in $H''$. From our choice of $T_i$, we have that $T_j = \lts{T_i}{x}{H''}$ and $\vli{T_j}{x}{H''} = vt_j \neq nil$. Thus $vt_j$ is not yet deleted in $H''$. Since $T_k$ has already read the value of $x$ from $T_j$, $k$ must be stored in the read list of $vt_j$. But when $T_i$ executes its \tryc{} \op, the algorithm will detect the condition $j<i<k$ from $vt_j$. So, $T_i$ will be aborted. Thus, $T_j = \lts{T_k}{x}{H_{gc}}$. 
}
\end{proof}

\cmnt{
\begin{proof}
Suppose, when a new read operation $r_i(x,v)$ occurs, it will read from $w_j(x,v)$ where $j<i$ but j is not the largest $closest\_id$. This means it selecting wrong $(j,v)$ from $<j,v,r_l>$ in the line no 2 of $Select(i,x)$ function. If this is true then some id $k$ of the transaction $T_k$ is the $closest\_id$ less than $i$ and greater than $j$, has been deleted. That is $j<k<i$.

Let $k'$ be the next version of $k$. Thus from the above two assumptions, we conclude that $j<k<i<k'$. So if we want to delete the id $k$, $i$ has to terminate before the deletion. But $T_i$ is reading from $T_j$, it means $i$ is in the read list of $j$ not in the $r_l$ of $k$. Thus $k$ can not be deleted which is contradictory. So $T_i$ has to read from $T_j$. Therefore $j$ is the closest tuple of $i$. That means $j \leqslant i$. Hence the history $H$ generated by MVTO with garbage collection satisfy the read rule.
\end{proof}
}

\begin{lemma}
\label{lem:write_rule}
The history $H_{gc}$ generated by MVTO with garbage collection satisfies write rule: Suppose a transaction $T_k$ reads $x$ from (a committed transaction) $T_j$ in $H_{gc}$, i.e. $\{w_j(x, v), r_k(x, v)\} \in \evts{H_{gc}}$. Let $T_i$ be a committed transaction that writes to $x$, i.e. $w_i(x, u) \in \evts{T_i}$. Then, the timestamp of $T_i$ is either less than $T_j$'s timestamp or greater than $T_k$'s timestamp, i.e. $i<j \oplus k<i$. 
\end{lemma}

\begin{proof}
Consider a sequential prefix of a $H_{gc}$, say $H'$, in which transaction $T_i$ has not yet executed \tryc{} \op{} but $T_k$ has read from $T_j$. Thus $T_i$ is a live a transaction in $H'$. From \lemref{gc-read-rule}, we have that $T_j = \lts{T_k}{x}{H'}$. Suppose by contradiction, $j<i<k$. This implies that $T_j = \lts{T_i}{x}{H'}$ as well. Thus, $\vli{T_j}{x}{H''} = vt_j \neq nil$ implying that $vt_j$ is not yet deleted in $H'$. The read function stores $T_k$ in $vt_j$. When $T_i$ executes \tryc{} \op, the algorithm detects $j<i<k$ and aborts $T_i$ which contradicts our assumption. 

Thus, we have that $i<j \oplus k<i$. 
\end{proof}

Since the read-rule and the write-rules are maintained, we get that \lemref{edgeorder} is true as well. Hence, \thmref{to-opaque} automatically follows. Thus the history generated by the algorithm with garbage collection is opaque as well.

\section{Conclusion}
\label{sec:conc}

There are many applications that require long running read-only transactions. Many STM systems can cause such transactions to abort. Multi-version STM system ensure that a read-only transactions does not need to abort by maintaining multiple versions. Two important properties that should be considered while building a STM system are: correctness which normally is opacity and progress condition which for multiversion systems is mv-permissiveness. Although several multi-version STM systems have been proposed to the best of our knowledge none of them have been proved formally satisfy opacity. 

In this paper we presented a timestamp based multiversion STM system that satisfies opacity and mv-permissiveness. We also presented an algorithm for garbage collection that deletes version that will never be used. We have formally proved the correctness of our algorithm including garbage collection.  

As a part of future work, we would like to implement this algorithm and test its performance on various benchmarks. Recently, Attiya and Hillel \cite{attiyaHill:sinmvperm:tcs:2012} proposed a single-version STM system that is mv-permissive. Their system uses Compare and Swap (CAS) primitives in addition to lock. As a part of the implementation, we would like to compare the performance of our algorithm with theirs to see how much benefit do multiple versions offer.

\bibliography{citations}

\begin{thebibliography}{10}

\bibitem{attiyaHill:sinmvperm:tcs:2012}
Hagit Attiya and Eshcar Hillel.
\newblock A single-version stm that is multi-versioned permissive.
\newblock {\em Theory Comput. Syst.}, 51(4):425--446, 2012.

\bibitem{AydAbd:2008:Serial:transact}
Utku Aydonat and Tarek Abdelrahman.
\newblock Serializability of transactions in software transactional memory.
\newblock In {\em TRANSACT~'08: 3rd Workshop on Transactional Computing}, feb
  2008.

\bibitem{BernGood:1983:MCC:TDS}
Philip~A. Bernstein and Nathan Goodman.
\newblock Multiversion concurrency control: Theory and algorithms.
\newblock {\em ACM Trans. Database Syst.}, 8(4):465--483, December 1983.

\bibitem{CachRito:2006:CP}
J.~Cachopo and A.~Rito-Silva.
\newblock {Versioned boxes as the basis for memory transactions}.
\newblock {\em Science of Computer Programming}, 63(2):172--185, 2006.

\bibitem{dice:2006:tl2:disc}
Dave Dice, Ori Shalev, and Nir Shavit.
\newblock Transactional locking {II}, booktitle = {DISC~'06: Proc. 20th
  International Symposium on Distributed Computing}, note = {Springer-Verlag
  Lecture Notes in Computer Science volume 4167}, pages = {194--208}, year =
  {2006}, month = {sep},.

\bibitem{Guer+:disc:2008}
Rachid Guerraoui, Thomas Henzinger, and Vasu Singh.
\newblock Permissiveness in transactional memories.
\newblock In {\em DISC~'08: Proc. 22nd International Symposium on Distributed
  Computing}, pages 305--319, sep 2008.
\newblock Springer-Verlag Lecture Notes in Computer Science volume 5218.

\bibitem{GuerKap:2008:PPoPP}
Rachid Guerraoui and Michal Kapalka.
\newblock On the correctness of transactional memory.
\newblock In {\em PPoPP '08: Proceedings of the 13th ACM SIGPLAN Symposium on
  Principles and practice of parallel programming}, pages 175--184, New York,
  NY, USA, 2008. ACM.

\bibitem{tm-book}
Rachid Guerraoui and Michal Kapalka.
\newblock {\em Principles of Transactional Memory,Synthesis Lectures on
  Distributed Computing Theory}.
\newblock Morgan and Claypool, 2010.

\bibitem{HerlMoss:1993:SigArch}
Maurice Herlihy and J.~Eliot B.Moss.
\newblock Transactional memory: architectural support for lock-free data
  structures.
\newblock {\em SIGARCH Comput. Archit. News}, 21(2):289--300, 1993.

\bibitem{herlihy+:2003:stm-dynamic:podc}
Maurice Herlihy, Victor Luchangco, Mark Moir, and III {William N. Scherer}.
\newblock Software transactional memory for dynamic-sized data structures.
\newblock In {\em PODC~'03: Proc. 22nd ACM Symposium on Principles of
  Distributed Computing}, pages 92--101, Jul 2003.

\bibitem{KuzSat:Corr:2012}
Petr Kuznetsov and Sathya Peri.
\newblock On non-interference of transactions.
\newblock {\em CoRR}, abs/1211.6315, 2012.

\bibitem{KR:2011:OPODIS}
Petr Kuznetsov and Srivatsan Ravi.
\newblock On the cost of concurrency in transactional memory.
\newblock In {\em OPODIS}, pages 112--127, 2011.

\bibitem{Papad:1979:JACM}
Christos~H. Papadimitriou.
\newblock The serializability of concurrent database updates.
\newblock {\em J. ACM}, 26(4):631--653, 1979.

\bibitem{Perel+:2011:SMV:DISC}
Dmitri Perelman, Anton Byshevsky, Oleg Litmanovich, and Idit Keidar.
\newblock Smv: Selective multi-versioning stm.
\newblock In {\em DISC}, pages 125--140, 2011.

\bibitem{perel+:2010:MultVer:PODC}
Dmitri Perelman, Rui Fan, and Idit Keidar.
\newblock On maintaining multiple versions in stm.
\newblock In {\em PODC}, pages 16--25, 2010.

\bibitem{riegel+:2006:LSA:DISC}
Torvald Riegel, Pascal Felber, and Christof Fetzer.
\newblock A lazy snapshot algorithm with eager validation.
\newblock In {\em Proceedings of the 20th International Symposium on
  Distributed Computing, {DISC} 2006}, volume 4167 of {\em Lecture Notes in
  Computer Science}, pages 284--298. Springer, Sep 2006.

\bibitem{riegel+:sistm:transact:2006}
Torvald Riegel, Christof Fetzer, and Pascal Felber.
\newblock Snapshot isolation for software transactional memory.
\newblock In {\em Proceedings of the First ACM SIGPLAN Workshop on Languages,
  Compilers, and Hardware Support for Transactional Computing}. Jun 2006.

\bibitem{ShavTou:1995:PODC}
Nir Shavit and Dan Touitou.
\newblock Software transactional memory.
\newblock In {\em PODC '95: Proceedings of the fourteenth annual ACM symposium
  on Principles of distributed computing}, pages 204--213, New York, NY, USA,
  1995. ACM.

\end{thebibliography}

\end{document}